\pdfoutput=1
\documentclass[11pt]{amsart}
\usepackage{tabularx,booktabs,tikz}
\usepackage{caption}
\usepackage{fancyhdr}
\usepackage{amsmath}
\usepackage{amsfonts}

\usepackage{amscd}
\usepackage{amsthm}
\usepackage{amssymb} \usepackage{latexsym}
\usepackage{mathrsfs}
\usepackage{eufrak}
\usepackage{euscript}
\usepackage{epsfig}
\usepackage{graphics}
\usepackage{array}
\usepackage{enumerate}
\usepackage{dsfont}
\usepackage{color}
\usepackage{wasysym}
\usepackage{hyperref}
\usepackage{pdfsync}
\allowdisplaybreaks[4]

\newcommand{\bel}[1]{\begin{equation}\label{#1}}

\newcommand{\be}{\begin{equation}}

\newcommand{\ba}{\begin{eqnarray}}
\newcommand{\ea}{\end{eqnarray}}

\newcommand{\qe}{\end{equation}}

\newcommand{\Hmm}[1]{\leavevmode{\marginpar{\tiny%
$\hbox to 0mm{\hspace*{-0.5mm}$\leftarrow$\hss}%
\vcenter{\vrule depth 0.1mm height 0.1mm width \the\marginparwidth}%
\hbox to
0mm{\hss$\rightarrow$\hspace*{-0.5mm}}$\\\relax\raggedright #1}}}

\newtheorem{theorem}{Theorem}[section]

\newtheorem{lemma}[theorem]{Lemma}
\newtheorem{corollary}[theorem]{Corollary}
\newtheorem{definition}[theorem]{Definition}

\newtheorem{remark}[theorem]{Remark}

\newtheorem{prop}[theorem]{Proposition}

\newcommand{\tm}{\begin{theorem}}
\newcommand{\tmd}{\end{theorem}}
\newcommand{\co}{\begin{corollary}}
\newcommand{\cod}{\end{corollary}}
\newcommand{\prp}{\begin{prop}}
\newcommand{\prpd}{\end{prop}}
\begin{document}

\title[topological solutions to the Chern-Simons model]{The existence of topological solutions to the Chern-Simons model on lattice graphs}

\author{Bobo Hua}
\address{Bobo Hua: School of Mathematical Sciences, LMNS, Fudan University, Shanghai 200433, China; Shanghai Center for Mathematical Sciences, Fudan University, Shanghai 200433, China}
\email{bobohua@fudan.edu.cn}

\author{Genggeng Huang}
\address{Genggeng Huang: School of Mathematical Sciences, LMNS, Fudan University, Shanghai 200433, China}
\email{genggenghuang@fudan.edu.cn}

\author{Jiaxuan Wang}
\address{Jiaxuan Wang: School of Mathematical Sciences, Fudan University, Shanghai 200433, China}
\email{jiaxuanwang21@m.fudan.edu.cn}

\begin{abstract}

We prove the existence of topological solutions to the self-dual Chern-Simons model and the Abelian Higgs system on the lattice graphs $\mathbb{Z}^n$ for $n\geqslant2$. This extends the results in \cite{huang2020existence} from finite graphs to lattice graphs.
\end{abstract}
\maketitle

%\tableofcontents
%\tableofcontents

\par
\maketitle

\bigskip

\section{Introduction}

Various vortex problems have been extensively studied in recent decades, which play important roles in quantum physics, solid state physics and so on. The existence of topological and non-topological solutions in these models has been rigorously proven in mathematics. In $\mathbb{R}^2$, we consider the self-dual Chern-Simons vortex equation
\begin{equation}
\begin{aligned}
\Delta u=\lambda e^u(e^u-1)+4\pi\sum_{j=1}^Mn_j\delta_{p_j}
\end{aligned}
\end{equation}
and the Abelian Higgs equation
\begin{equation}
\begin{aligned}
\Delta u=\lambda (e^u-1)+4\pi\sum_{j=1}^Mn_j\delta_{p_j}
\end{aligned}
\end{equation}
with positive integers $n_1,\ldots,n_M$ and distinct vortices $p_1, \ldots, p_M\in \mathbb{R}^2$. Here $\lambda>0$, and $\delta_{p_j}$ is the Dirac mass at $p_j$. A solution of $(1)$ or $(2)$ is called topological if $u(x)\rightarrow0$ as $|x|\rightarrow+\infty$, and called non-topological if $u(x)\rightarrow-\infty$ as $|x|\rightarrow+\infty$.

For the Abelian Higgs system $(2)$, Jaffe and Taubes proved the existence and uniqueness of general finite energy multivortex solutions to the Bogomol'nyi equations in \cite{jaffe1980vortices}, and there have been many studies on this model since then, such as \cite{jacobs1979interaction,jaffe1980vortices,wang1992abrikosov}. For the self-dual Chern-Simons system $(1)$, it is the minimal self-dual model containing the Chern-Simons term. The Chern-Simons vortices were discovered in \cite{jackiw1990self,hong1990multivortex}, which attracted people to investigate the existence problem. The existence of topological solutions in $\mathbb{R}^2$ was established in \cite{ronggang1991existence,spruck1995topological} by the variational method and iteration argument, and the existence of self-dual doubly periodic vortex solutions was proved in \cite{caffarelli1995vortex}. Later, non-topological solutions were studied in the literature, e.g. \cite{chen1994nonlinear,chae2000existence,chan2002non,choe2011existence}, and see \cite{dunne2009self,han2014existence,struwe1998multivortex,chae1997topological} for other related results.

Recently, people have paid attention to the elliptic equations on graphs. Grigor'yan, Lin and Yang first studied nonlinear elliptic equations on graphs, see e.g. \cite{grigor2016yamabe,grigor2017existence}. In a seminal paper \cite{huang2020existence}, Huang, Lin and Yau proved the existence result for solutions to $(1)$ on finite graphs. Furthermore, on a finite graph, the existence of solutions to the generalized self-dual Chern-Simons equation was proved in \cite{lu2021existence,hou2022existence}, and the existence of solutions to the Chern-Simons Higgs model has been proved in \cite{li2023topological} using topological degree methods recently. See \cite{grigor2016kazdan,huang2021mean,chao2023multiple} for other related results. For infinite graphs, existence results for the Kazdan-Warner equation were proved in \cite{article,article1} on graphs with positive spectrum and canonically compactifiable graphs, while lattice graphs $\mathbb{Z}^n$, i.e. discrete analogs of Euclidean spaces $\mathbb{R}^n$, are excluded. In this paper, our main contribution is to extend the results in \cite{huang2020existence} from finite graphs to lattice graphs. We study the Chern-Simons equation $(1)$ on $\mathbb{Z}^n$ for $n\geqslant2$, and prove the existence of topological solutions. Furthermore, using topological solutions of $(1)$, we prove the existence of topological solutions to the Abelian Higgs equation (2) on $\mathbb{Z}^n$ for $n\geqslant2$.

%We recall some settings of graphs. Let $G=(V,E)$ be a locally finite, simple, undirected and connected graph, consisting of the set of vertices $V$ and the set of edges $E$. We denote by $|G|$ the number of vertices in $V$. We call two vertices $x,y$ neighbours, denoted by $x\sim y$, if there exists $e\in E$ connecting $x$ and $y$. For any $x,y\in G$, the distance between them is defined as
%$$d(x,y)=\inf\{k:x=x_0\sim\ldots\sim x_k=y\}.$$
%The Laplacian on $G$ is defined as
%$$\Delta f(x) = \sum_{y\sim x}(f(y)-f(x)),\ f:G\rightarrow \mathbb{C}.$$
%For any finite subset $\Omega\subset V$, we define the boundary of $\Omega$ by
%$$\delta\Omega:=\{y\in V\setminus \Omega:\exists x\in\Omega,\ s.t.\ y\sim x\},$$
%and $\overline{\Omega}=\Omega\cup \delta\Omega$.

We consider the infinite integer lattice $\mathbb{Z}^n$ for $n\geqslant 2$, see Section~\ref{sec:1} for details. We define the distance on $\mathbb{Z}^n$ by
$$d(x,y)=\sum_{i=1}^n|x_i-y_i|,\ x,y\in\mathbb{Z}^n,$$
and write $d(x)=d(x,0)$. For any function $u:\mathbb{Z}^n\rightarrow \mathbb{R}$, the $l^p$-norm of $u$ is defined as
\begin{equation*}
\|u\|_{l^p(\mathbb{Z}^n)}=\left\{
\begin{aligned}
& \left(\sum_{x\in\mathbb{Z}^n}|u(x)|^p\right)^{\frac{1}{p}},\ 1\leqslant p<\infty,\\
& \sup_{x\in\mathbb{Z}^n}|u(x)|, \ p=\infty.\\
\end{aligned}
\right.
\end{equation*}
The Laplacian is defined as
$$\Delta u(x) = \sum_{d(x,y)=1}(u(y)-u(x)).$$
In the following we mainly consider topological solutions to the self-dual Chern-Simons vortex equation on $\mathbb{Z}^n$
\begin{equation*}
\left\{
\begin{aligned}
& \Delta u=\lambda e^u(e^u-1)+4\pi\sum_{j=1}^Mn_j\delta_{p_j}\ \ \text{on} \ \mathbb{Z}^n,\\
& \lim_{d(x)\rightarrow+\infty}u(x)=0,\\
\end{aligned}
\right.
\end{equation*}
and construct a topological solution to the above equation, which is maximal among all possible solutions. Our main result is as follows.

%\tm\label{thm:main1}
%Let $n\geqslant2,\ \lambda>0$ and $\Omega_0$ be a finite subset of $\mathbb{Z}^n$ containing the distinct points $\{p_j\}_{j=1}^M$. For any finite subset $\Omega\supset\Omega_0$, the boundary value problem
%\begin{equation*}
%\left\{
%\begin{aligned}
%& \Delta u=\lambda e^u(e^u-1)+4\pi\sum_{j=1}^Mn_j\delta_{p_j}\ \ \text{on} \ \Omega,\\
%& u(x) =0\ \ \text{on} \ \delta\Omega\\
%\end{aligned}
%\right.
%\end{equation*}
%has a solution $u_{\Omega}:\overline{\Omega}\rightarrow\mathbb{R}$. This solution is maximal among all possible solutions and satisfies that $\|u_\Omega\|_{l^p(\Omega)}\leqslant C_0$, where $1\leqslant p\leqslant\infty$, and $C_0$ only depends on dimension $n$, $p$, $\lambda$ and $\sum_{j=1}^Mn_j$.
%\tmd
\tm\label{thm:main2}
The equation $(1)$ has a topological solution $u\in l^p(\mathbb{Z}^n)$ on $\mathbb{Z}^n$ for $1\leqslant p\leqslant\infty$ and $n\geqslant 2$, which is maximal among all possible solutions. Furthermore, we have the decay estimate
$$u= O(e^{-m(1-\epsilon) d(x)}),$$
where $m=\ln(1+\frac{\lambda}{2n})$, $0<\epsilon<1$.
\tmd

In this paper, we provide two proofs of Theorem~\ref{thm:main2}. In the proof A, we adopt the exhaustion method and the discrete isoperimetric inequality. The proof A is novel, which relies on the discrete nature of graphs in an essential way. First, by a  contradiction argument, we prove the existence of solutions on a finite subset $\Omega$ with Dirichlet boundary condition in Lemma~\ref{lm4}. In order to prove the existence result on $\mathbb{Z}^n$, we apply the exhaustion method. This approach was first introduced by Lin-Yang to the analysis on graphs in \cite{lin2022calculus}. Considering a sequence of finite subsets
$$\Omega_0\subset \Omega_1\subset\ldots\subset \Omega_k\subset\ldots,\ \ \bigcup_{i=1}^\infty\Omega_i=\mathbb{Z}^n,$$
and corresponding monotone solution sequence $\{u_{\Omega_i}\}$ obtained by Lemma~\ref{lm4}, we denote by $\widetilde{u^i}(x)$ the null extension to $\mathbb{Z}^n$ of $u_{\Omega_i}$, see Subsection~\ref{subsec:1}. By passing to the limit, to avoid the triviality of the limit, one needs to prove a uniform bound for all $\widetilde{u^i}$. Suppose that it is not true, then there exists $\lim_{i\rightarrow+\infty}\widetilde{u^i}(x_i)=-\infty$ for a vertex sequence $\{x_i\}$. Set
$$A_1^i=\{x\in\Omega_i:\ -C\leqslant \widetilde{u^i}(x)\leqslant0\},$$
$$A_2^i=\{x\in\Omega_i:\ \widetilde{u^i}(x)<-(2n+1)C-\lambda\},$$
$$A_3^i=\{x\in\Omega_i:\ -(2n+1)C-\lambda\leqslant\widetilde{u^i}(x)<-C\},$$
where $C=4\pi\sum_{j=1}^Mn_j$. Since $|\Delta\widetilde{u^i}(x)|$ is bounded on $\Omega_i$, we may prove that $A_3^i\neq\emptyset$ and $\lim_{i\rightarrow+\infty}|A_2^i|=+\infty$. These imply that $\lim_{i\rightarrow+\infty}|A_3^i|=+\infty$ by the isoperimetric inequality on $\mathbb{Z}^n$. However, summing over $\Omega_i$ in the equation, we know that $\sum_{x\in\Omega_i}e^{\widetilde{u^i}}(1-e^{\widetilde{u^i}})$ is uniformly bounded, which yields a contradiction. Hence we prove the $l^\infty-$convergence $\widetilde{u^i}\rightarrow u$ on $\mathbb{Z}^n$, and this limit is a topological solution. Applying the maximum principle and Lemma~\ref{lm5}, we finally get the decay estimate and the maximality of the constructed solution.

In the proof B, we follow the methods in \cite{spruck1995topological}. We prove a key lemma, Lemma~\ref{norm}, which provides the uniform $l^2$-norm estimate of the solution on a finite subset $\Omega$ with Dirichlet boundary condition. To prove this lemma, let $F(u)$ be the natural functional associated to the equation $(1)$ on $\Omega$. By Green's identities on graphs, we prove that $F(u_k)$ decreases with respect to $k$ and has a upper bound which only depends on $n$, $\lambda$ and $\sum_{j=1}^Mn_j$. Applying the discrete Gagliardo-Nirenberg-Sobolev inequality proved in \cite{porretta2020note}, we have
$$\|u_k\|_{l^2(\Omega)}\leqslant C_3(F(u_k)+1)\leqslant C_4,$$
where $C_3,C_4$ only depend on $n$, $\lambda$ and $\sum_{j=1}^Mn_j$. Thanks to this lemma, we may pass to the limit and get the solution $u_\Omega\in l^2(\Omega)$ on $\Omega$. Since its $l^2$-norm is uniformly bounded, we construct the solution $u$ of the equation $(1)$ on $\mathbb{Z}^n$ by the exhaustion method. As in the proof A, we prove the decay estimate and the maximality of the solution.

With the help of Theorem~\ref{thm:main2}, we prove the existence of topological solutions to the Abelian Higgs model.

\tm\label{thm:main3}
The equation $(2)$ has a unique topological solution $u'\in l^p(\mathbb{Z}^n)$ on $\mathbb{Z}^n$ for $1\leqslant p\leqslant\infty$ and $n\geqslant 2$, satisfying $u\leqslant u'\leqslant0$, where $u$ is constructed in Theorem~\ref{thm:main2}. Furthermore, there holds the decay estimate
$$u'= O(e^{-m(1-\epsilon) d(x)}),$$
where $m=\ln(1+\frac{\lambda}{2n})$, $0<\epsilon<1$.
\tmd

To prove Theorem~\ref{thm:main3}, we apply the sub-supersolution method. By choosing $\omega_1=0$ as a supersolution and $\omega_2=u$ as a subsolution, where
$u$ is constructed in Theorem~\ref{thm:main2}, we obtain a solution to (2) by the monotone iteration argument.

The paper is organized as follows:
In next section, we introduce the setting of graphs. In Section~\ref{sec:2}, we give two proofs of Theorem~\ref{thm:main2}. In Section~\ref{sec:3}, we prove Theorem~\ref{thm:main3}.

\textbf{Acknowledgements.} B.H. is supported by NSFC, No. 11831004. All authors of the paper are supported by Shanghai Science and Technology Program [Project No. 22JC1400100].

\section{Preliminaries}\label{sec:1}

\subsection{The setting of $\mathbb{Z}^n$}\label{subsec:1} \

\

Consider the infinite integer lattice graph $\mathbb{Z}^n$, $n\geqslant 2$, consisting of the set of vertices
$$V=\mathbb{Z}^n=\{x=(x_1,\ldots,x_n)\in\mathbb{R}^n:x_i\in \mathbb{Z},\forall 1\leqslant i\leqslant n\}$$
and the set of edges
$$E=\{\{x,y\}:x,y\in\mathbb{Z}^n,\sum_{i=1}^n|x_i-y_i|=1\},$$
and we write $x\sim y$ if $\{x,y\}\in E$. We denote by $C(\mathbb{Z}^n)=\{u:\mathbb{Z}^n\rightarrow\mathbb{R}\}$ the set of functions on $\mathbb{Z}^n$, and by $supp(u)=\{x\in\mathbb{Z}^n:u(x)\neq0\}$ the support of $u$, and by $C_0(\mathbb{Z}^n)$ the set of functions with finite support.
For a finite subset $\Omega\subset\mathbb{Z}^n$, we define the boundary of $\Omega$ as
$$\delta\Omega:=\{y\in \mathbb{Z}^n\setminus \Omega:\exists x\in\Omega\ \text{such that}\ y\sim x\},$$
and write $\overline{\Omega}=\Omega\cup \delta\Omega$. For $u\in C(\Omega)$, the null extension to $\mathbb{Z}^n$ of $u$ is defined as
\begin{equation*}
\widetilde{u}(x)=\left\{
\begin{aligned}
& u(x) &\text{on} \ \Omega,\\
& 0 &\text{on} \ \Omega^c.\\
\end{aligned}
\right.
\end{equation*}

We define the difference operator as
$$\nabla_{xy}u=u(y)-u(x),\ u\in C(\mathbb{Z}^n),\ x,y\in\mathbb{Z}^n.$$
For $f,g\in C(\overline{\Omega})$, we introduce a bilinear form
$$D_{\Omega}(f,g):=\frac{1}{2}\sum_{\substack{x,y\in \Omega\\ x\sim y}}\nabla_{xy}f\nabla_{xy}g+\sum_{\substack{x\in \Omega,y\in \delta\Omega\\ x\sim y}}\nabla_{xy}f\nabla_{xy}g,$$
and we write $D_\Omega(f)=D_\Omega(f,f)$ for the Dirichlet energy of $f$ on $\Omega$. For $f\in C(\overline{\Omega})$, the directional derivative operator $\frac{\partial f}{\partial \vec{n}}$ at $x\in \delta\Omega$ is defined as
$$\frac{\partial f}{\partial \vec{n}}(x):=\sum_{\substack{y\in \Omega\\ x\sim y}}(f(x)-f(y)).$$

The following are Green's identities on graphs, see e.g. \cite{grigor2018introduction}.
\begin{lemma}\label{lm1}
Let $f,g\in C(\mathbb{Z}^n)$ and $\Omega$ be a finite subset of $\mathbb{Z}^n$.
\begin{enumerate}[(a)]
\item If $f\in C_0(\mathbb{Z}^n)$, we have
$$\frac{1}{2}\sum_{\substack{x,y\in \mathbb{Z}^n\\ x\sim y}}\nabla_{xy}f\nabla_{xy}g=-\sum_{x\in\mathbb{Z}^n}f(x)\Delta g(x).$$
\item $$D_\Omega(f,g)=-\sum_{x\in\Omega}f(x)\Delta g(x)+\sum_{x\in\delta\Omega}f(x)\frac{\partial g}{\partial \vec{n}}(x).$$
\end{enumerate}
\end{lemma}

\

\subsection{Maximum principle and discrete functional inequalities} \

\

In this subsection we introduce a maximum principle, the isoperimetric inequality, and the discrete Gagliardo-Nirenberg-Sobolev inequality on $\mathbb{Z}^n$, which play key roles in the proofs of main results. The following maximum principle is well-known.
\begin{lemma}\label{lm2}
Let $\Omega$ be a finite subset of $\mathbb{Z}^n$. For any positive $f\in C(\overline{\Omega})$, suppose that a function $v\in C(\overline{\Omega})$ satisfies
\begin{equation*}
\left\{
\begin{aligned}
& (\Delta-f)v\geqslant 0\ \ \text{on}\ \Omega,\\
& v\leqslant0\ \ \ \ \ \ \quad \ \ \text{on}\ \delta\Omega .\\
\end{aligned}
\right.
\end{equation*}
We have $v\leqslant0$ on $\overline{\Omega}$.
\end{lemma}

\begin{proof}

We prove the result by contradiction. Suppose that there exists $x\in \Omega$ such that $v(x)=\sup_{y\in\overline{\Omega}}v(y)=c>0$. By the equation,
$$\Delta v(x)\geqslant f(x)v(x)>0.$$
This implies that there exists $x_0\sim x,\ x_0\in\overline{\Omega}$, such that $v(x_0)>v(x)=c,$
which yields a contradiction.
\end{proof}

By the above lemma, we have the following corollary.

\co\label{co1}
For any positive $f\in C(\mathbb{Z}^n)$, suppose that a function $v\in l^2(\mathbb{Z}^n)$ satisfies
\begin{equation*}
\left\{
\begin{aligned}
& (\Delta-f)v\geqslant 0\ \ \text{on}\ \mathbb{Z}^n,\\
& \lim_{d(x)\rightarrow+\infty}v(x)\leqslant0 .\\
\end{aligned}
\right.
\end{equation*}
Then $v\leqslant0$ on $\mathbb{Z}^n$.
\cod
\

The isoperimetric inequality is well-known on $\mathbb{Z}^n$, see e.g. \cite{barlow2017random}, which is needed for our proof A. For $K\subset\mathbb{Z}^n$, we denote by $|K|$ the cardinality of the set $K$.

\begin{lemma}\label{II}
There exists a constant $C_n$, only depending on the dimension $n$, such that for any finite $\Omega\subset\mathbb{Z}^n$,
$$|\delta\Omega|\geqslant C_n|\Omega|^{\frac{n-1}{n}}.$$
\end{lemma}

For $p\geqslant1$, we define the $D^{1,p}$ norm as
$$\|u\|_{D^{1,p}(\mathbb{Z}^n)}:=\left(\sum_{x\in\mathbb{Z}^n}\sum_{y\sim x}|u(y)-u(x)|^p\right)^{\frac{1}{p}}.$$
In the proof B, we need the discrete Gagliardo-Nirenberg-Sobolev inequality on $\mathbb{Z}^n$. Since $\mathbb{Z}^n$ is a discrete regular mesh, the proof of Theorem 4.1 in \cite{porretta2020note} yields the following discrete Gagliardo-Nirenberg-Sobolev inequality.

\begin{lemma}\label{GN}$($\cite{porretta2020note}$)$
Let $n\geqslant2$, $p>1$, $\gamma\geqslant p$ and $p'=\frac{p}{p-1}$. Then for any $u\in l^p(\mathbb{Z}^n)$, we have
$$\|u\|^{\gamma}_{l^{\frac{\gamma n}{n-1}}(\mathbb{Z}^n)}\leqslant C(p,n,\gamma)\|u\|_{D^{1,p}(\mathbb{Z}^n)}\|u\|^{\gamma-1}_{l^{(\gamma-1)p'}(\mathbb{Z}^n)}.$$
\end{lemma}

\begin{remark}
Although Theorem 4.1 in \cite{porretta2020note} requires $p>n$ and $\gamma> p$, the above inequality in fact holds for any $p>1$ and $\gamma\geqslant p$ by the same argument in \cite{porretta2020note}. For $n\geqslant2$, choose
$$p=\gamma=2,\ p'=2.$$
With a well-known fact that for any $q\geqslant p$,
$$\|u\|_{l^q(\mathbb{Z}^n)}\leqslant\|u\|_{l^p(\mathbb{Z}^n)},$$
see e.g. in \cite{GLY}, we get for $u\in l^2(\mathbb{Z}^n)$,
$$\|u\|_{l^4(\mathbb{Z}^n)}\leqslant \|u\|_{l^{\frac{2n}{n-1}}(\mathbb{Z}^n)}\leqslant C_n' \|u\|^{\frac{1}{2}}_{D^{1,2}(\mathbb{Z}^n)}\|u\|^{\frac{1}{2}}_{l^2(\mathbb{Z}^n)}.$$

\end{remark}

\section{Existence theorems for the Chern-Simons equation}\label{sec:2}

In this section we consider the existence of topological solutions to the Chern-Simons equation, and we give two proofs of Theorem~\ref{thm:main2}. To prove this theorem, we first consider an iterative sequence on a finite subset of $\mathbb{Z}^n$.

Let $\Omega_0$ be a finite subset of $\mathbb{Z}^n$, satisfying $\Omega_0\supset \{p_j\}_{j=1}^M$, and $\Omega$ be an arbitrary connected finite subset such that $\Omega_0\subset\Omega\subset\mathbb{Z}^n$. We write
$$g=4\pi\sum_{j=1}^Mn_j\delta_{p_j},\ C=4\pi\sum_{j=1}^Mn_j,$$
and it is obvious that $g\in l^p(\mathbb{Z}^n)$ for any $p\geqslant1$. Choose a constant $K>2\lambda>0$. Let $u_0=0$ and consider the following iterative equations,
\begin{equation}
\left\{
\begin{aligned}
& (\Delta-K) u_k=\lambda e^{u_{k-1}}(e^{u_{k-1}}-1)+g-Ku_{k-1}\ \ \text{on} \ \Omega,\\
& u_k =0\ \ \text{on} \ \delta\Omega.\\
\end{aligned}
\right.
\end{equation}

\begin{lemma}\label{lm3}
Let the sequence $\{u_k\}$ be given in $(3)$. Then for each $k$, $u_k$ is uniquely defined and
$$0=u_0\geqslant u_1\geqslant u_2\geqslant\ldots.$$
\end{lemma}

\begin{proof}
First we have
\begin{equation}
\left\{
\begin{aligned}
& (\Delta-K) u_1=g\ \ \text{on} \ \Omega,\\
& u_1 =0\ \ \quad\quad\quad\ \text{on} \ \delta\Omega.\\
\end{aligned}
\right.
\end{equation}
One easily sees the existence and uniqueness of the solution $u_1$ on $\Omega$. Using Lemma~\ref{lm2}, we obtain that $u_1\leqslant0$.

Suppose that $0=u_0\geqslant u_1\geqslant u_2\geqslant\ldots\geqslant u_{i}$. Since
$$\lambda e^{u_i}(e^{u_i}-1)+g-Ku_i\in l^2(\Omega),$$
we have the existence and uniqueness of the solution $u_{i+1}$. From the equations $(3)$, we get
\begin{equation*}
\begin{aligned}
(\Delta-K)(u_{i+1}-u_i)&=\lambda(e^{2u_i}-e^{2u_{i-1}})-\lambda(e^{u_i}-e^{u_{i-1}})-K(u_i-u_{i-1})\\
&\geqslant2\lambda e^{2\omega}(u_i-u_{i-1})-K(u_i-u_{i-1})\\
&\geqslant K(e^{2\omega}-1)(u_i-u_{i-1})\geqslant 0,
\end{aligned}
\end{equation*}
where $\omega$ is a function satisfying $u_i\leqslant\omega\leqslant u_{i-1}$. This implies that $u_{i+1}\leqslant u_i$ by Lemma~\ref{lm2} and proves this lemma.
\end{proof}

\

\subsection{The proof A of Theorem~\ref{thm:main2}}

\

By Lemma~\ref{lm3}, we prove the convergence of the monotone sequence $\{u_k\}$.

\begin{lemma}\label{lm4}
Let $\{u_k\}$ be the sequence defined by $(3)$. Then there exists $u_\Omega\in C(\overline{\Omega})$ such that
$$u_k\rightarrow u_\Omega \ \text{on}\ \overline{\Omega},$$
which satisfies
\begin{equation}
\left\{
\begin{aligned}
& \Delta u_\Omega=\lambda e^{u_\Omega}(e^{u_\Omega}-1)+g \ \ \text{on} \ \Omega,\\
& u_\Omega =0\ \ \text{on} \ \delta\Omega.\\
\end{aligned}
\right.
\end{equation}
\end{lemma}

\

\begin{proof}
Since $\Omega$ is finite and the sequence is monotone, the pointwise limit $u_\Omega$ of $u_k$ exists. It suffices to show that $u_\Omega$ is bounded. We first consider the set
$$B(\Omega)=\{x\in \Omega: \exists y\in \delta\Omega\ \text{such that}\ y\sim x\}.$$
Summing over $\Omega$ in (3), and by Lemma~\ref{lm1} we obtain
\begin{equation*}
\begin{aligned}
&\sum_{x\in \delta\Omega}\frac{\partial u_k}{\partial \vec{n}}(x)+\lambda\sum_{x\in\Omega}e^{u_k}(1-e^{u_k})\\
=&\sum_{x\in\Omega}g(x)+K\sum_{x\in\Omega}(u_k(x)-u_{k-1}(x))\leqslant4\pi\sum_{j=1}^Mn_j=C.\\
\end{aligned}
\end{equation*}
This yields that
$$\sum_{x\in B(\Omega)}|u_k(x)|\leqslant C.$$
In particular, for any $x\in B(\Omega)$, the sequence $\{u_k(x)\}$ is uniformly bounded.

If $x_1\sim x_0$, $x_1\in\Omega$ and $x_0\in B(\Omega)$, we claim that $\{u_k(x_1)\}$ is uniformly bounded. The equation (3) at $x_0$ shows that
\begin{equation*}
\begin{aligned}
|\Delta u_k(x_0)|&\leqslant K|u_k(x_0)-u_{k-1}(x_0)|+\lambda|e^{u_{k-1}}(e^{u_{k-1}}-1)|+|g(x)|\\
&\leqslant K|u_k(x_0)|+\frac{\lambda}{4}+C\leqslant (K+1)C+\frac{\lambda}{4}.
\end{aligned}
\end{equation*}
Note that
$$\Delta u_k(x_0)=\sum_{y\sim x_0}(u_k(y)-u_k(x_0))\leqslant u_k(x_1)-2nu_k(x_0)\leqslant u_k(x_1)+2nC$$
and $u_k(x_1)<0$. We obtain that $\{u_k(x_1)\}$ is uniformly bounded.

Since $\Omega$ is connected, we repeat the above process, and get $\{u_k(x)\}$ is uniformly bounded on $\Omega$, which completes the proof of this lemma.
\end{proof}

\

Let $\Omega_i$, $1\leqslant i<\infty$, be finite and connected subsets, satisfying
$$\Omega_0\subset \Omega_1\subset\ldots\subset \Omega_k\subset\ldots,\ \ \bigcup_{i=1}^\infty\Omega_i=\mathbb{Z}^n.$$
We write $u^i =u_{\Omega_i}$ for simplicity. To prove Theorem~\ref{thm:main2}, we need the following lemma.
\begin{lemma}\label{lm5}
Let $\Omega$ be a finite subset of $\mathbb{Z}^n$ and $\{u_k\}$ be the sequence defined by $(3)$. For any function $V\in C(\overline{\Omega})$ satisfying
\begin{equation*}
\left\{
\begin{aligned}
& \Delta V\geqslant\lambda e^V(e^V-1)+g\ \ \text{on} \ \Omega,\\
& V(x) \leqslant0\ \ \text{on} \ \delta\Omega,\\
\end{aligned}
\right.
\end{equation*}
we have
$$0=u_0\geqslant u_1\geqslant\ldots\geqslant u_k\geqslant\ldots\geqslant u_\Omega\geqslant V.$$
\end{lemma}

\begin{proof}
First, one has
$$\Delta V\geqslant\lambda e^V(e^V-1)+g\geqslant\lambda e^V(e^V-1).$$
We claim that $\sup_{x\in\Omega}V(x)\leqslant0$. If not, choose $V(x_0)=\sup_{x\in\Omega}V(x)>0$ for some $x_0\in \Omega$. Then
$$0\geqslant\Delta V(x_0)\geqslant \lambda e^{V(x_0)}(e^{V(x_0)}-1)>0,$$
which yields a contradiction and proves the claim.

Suppose that $V\leqslant u_k$, then
\begin{equation*}
\begin{aligned}
(\Delta-K)(u_{k+1}-V)&\leqslant \lambda(e^{2u_k}-e^{2V})-\lambda(e^{u_k}-e^{V})-K(u_k-V)\\
&\leqslant K (e^{2\omega}-1)(u_k-V)\leqslant 0,
\end{aligned}
\end{equation*}
where the function $\omega$ satisfies $V\leqslant \omega\leqslant u_k\leqslant0$. This implies that $V\leqslant u_{k+1}$ by Lemma~\ref{lm2} and proves this lemma by the induction.
\end{proof}

\

Finally, we use these lemmas to prove Theorem~\ref{thm:main2}.
\begin{proof}[Proof of Theorem~\ref{thm:main2}]
For any integers $1\leqslant j\leqslant k$, we have $\Omega_j\subset \Omega_k$. On $\overline{\Omega_j}$, since $u^k\leqslant0$, one easily sees that $u^k$ satisfies the conditions in Lemma~\ref{lm5}, and we obtain
$$u^k\leqslant u^j \ \ \text{on}\ \overline{\Omega_j}.$$
Let $\widetilde{u^i}$ be the null extension to $\mathbb{Z}^n$ of $u^i$ on $\Omega_i$.
Note that
$$0\geqslant \widetilde{u^1}\geqslant\widetilde{u^2}\geqslant\ldots\geqslant\widetilde{u^k}\geqslant\ldots$$
on $\mathbb{Z}^n$.

In order to prove the convergence of the sequence $\{\widetilde{u^i}\}$ to a nontrivial limit, we need to give a uniform lower bound of $\{\widetilde{u^i}\}$. We argue by contradiction. Suppose that
$$\lim_{i\rightarrow+\infty}\|\widetilde{u^i}\|_{l^\infty}=+\infty.$$
We denote by
$$A_1^i=\{x\in\Omega_i:\ -C\leqslant \widetilde{u^i}(x)\leqslant0\},$$
$$A_2^i=\{x\in\Omega_i:\ \widetilde{u^i}(x)<-(2n+1)C-\lambda\},$$
$$A_3^i=\{x\in\Omega_i:\ -(2n+1)C-\lambda\leqslant\widetilde{u^i}(x)<-C\}.$$
By the conditions we assume, we get $A_2^i\neq\emptyset$ when $i\geqslant i_0$ for some $i_0$. In the following, we only consider $i\geqslant i_0$. Following the proof of Lemma~\ref{lm4}, we have
$$\sum_{x\in B(\Omega_i)}|\widetilde{u^i}(x)|\leqslant C,$$
which yields $B(\Omega_i)\subset A_1^i$ so that $A_1^i\neq \emptyset$. To obtain the uniform $l^\infty$-norm, we show the contradiction in three steps.

$(i)$: We claim that
$$A_3^i\neq\emptyset.$$
Suppose that $A_3^i=\emptyset$, which is equivalent to $A_1^i\cup A_2^i=\Omega_i$, and there exist two vertices $x,y\in \Omega_i$, satisfying
$x\sim y,\ x\in A_1^i,\ y\in A_2^i$. Note that
\begin{equation*}
\begin{aligned}
&\Delta \widetilde{u^i}(x)=\sum_{z\sim x}(\widetilde{u^i}(z)-\widetilde{u^i}(x))\leqslant \ \widetilde{u^i}(y)-2n\widetilde{u^i}(x)\\
&< -(2n+1)C-\lambda+2nC=-C-\lambda,
\end{aligned}
\end{equation*}
and
\begin{equation*}
\begin{aligned}
&|\Delta \widetilde{u^i}(x)|\leqslant |g(x)|+\lambda|e^{\widetilde{u^i}}(1-e^{\widetilde{u^i}})|<C+\lambda.
\end{aligned}
\end{equation*}
This yields a contradiction. Thus we have $A_3^i\neq\emptyset$.

$(ii)$: We claim that
$$\lim_{i\rightarrow+\infty}|A_2^i|=+\infty.$$
By
$$\lim_{i\rightarrow+\infty}\|\widetilde{u^i}\|_{l^\infty}=+\infty,$$
we choose a sequence $\{x_i\}$, where $x_i\in\Omega_i$, satisfying
$$\lim_{i\rightarrow+\infty}\widetilde{u^i}(x_i)=-\infty.$$
Consider the function
$$w_i(x)=\widetilde{u^i}(x-x_1+x_i),$$
and
$$\Omega'_i=\{x\in\mathbb{Z}^n:\ x-x_1+x_i\in\Omega_i\}.$$
We have
$$\lim_{i\rightarrow+\infty}w_i(x_1)=-\infty.$$
Let $i$ be sufficiently large. From the proof in $(i)$ we also have the estimate that for any $x\in \Omega'_i$,
$$|\Delta w_i(x)|<C+\lambda.$$
Consider an arbitrary vertex $y_1\in \Omega'_i$, $y_1\sim x_1$, and by the above facts,
\begin{equation*}
\begin{aligned}
\Delta w_i(y_1)=\sum_{z\sim y_1}(w_i(z)&-w_i(y_1))
\leqslant -2nw_i(y_1)+w_i(x_1).\\
\end{aligned}
\end{equation*}
This implies that
$$w_i(y_1)< \frac{1}{2n}w_i(x_1)+\frac{C+\lambda}{2n}.$$
Since $\Omega'_i$ is connected, $x_1\in\Omega'_i$ and
$$\lim_{i\rightarrow+\infty}|\Omega'_i|=+\infty,$$
there exist $y_1\sim x_1$, and a subsequence, still denoted by $\{\Omega'_i\}$, satisfying $y_1\in \Omega'_i$. One obtains that
$$\liminf_{i\rightarrow+\infty}w_i(y_1)=-\infty.$$
Repeating the above process,
$$\limsup_{i\rightarrow+\infty}|\{y\in \Omega'_i:\liminf_{i\rightarrow+\infty}w_i(y)=-\infty\}|=+\infty.$$
The monotonically decreasing sequence $\{\widetilde{u^i}\}$ guarantees that
$$A^i_2\subset A^{i+1}_2.$$
By letting $i\rightarrow+\infty$, one easily sees that
$$\lim_{i\rightarrow+\infty}|A_2^i|=\limsup_{i\rightarrow+\infty}|A_2^i|=+\infty.$$

$(iii)$: From $(i),(ii)$, we want to prove that
$$\limsup_{i\rightarrow+\infty}|A_3^i|=+\infty.$$
We argue by contradiction.

Suppose that
$$\limsup_{i\rightarrow+\infty}|A_3^i|= N<+\infty.$$
We focus on the set $\Omega_i\setminus A_3^i=A_1^i\cup A_2^i$, which can be divided into a union of several disjoint connected subsets, i.e.
$$\Omega_i\setminus A_3^i=\bigcup_{j=1}^lO_j,$$
and we have
$$\delta O_j\subset \delta\Omega_i\cup A_3^i.$$

From the proof of $(i)$, we have
$$O_j\subset A_1^i\ \text{or}\ O_j\subset A_2^i.$$
For some $1\leqslant l_1\leqslant l-1$, without loss of generality, we may assume that
$$A_1^i=\bigcup_{j=1}^{l_1}O_j,A_2^i= \bigcup_{j=l_1+1}^lO_j.$$
Since $B(\Omega_i)\subset A_1^i$, we get
$$\delta\Omega_i\subset \bigcup_{j=1}^{l_1}\delta O_j.$$
Thus for $l_1+1\leqslant j\leqslant l$,
$$\delta O_j\subset A_3^i.$$
For any $x\in\Omega_i$, since $O_1,\ldots,O_l$ are disjoint connected sets and $|\delta \{x\}|=2n$, there are no more than $2n$ sets from the family $\{O_j\}_{j=1}^l$ satisfying $x\in \delta O_j$.

By the isoperimetric inequality in Lemma~\ref{II}, we have the following estimate
\begin{equation*}
\begin{aligned}
|A_2^i|&= \sum_{j=l_1+1}^l|O_j|\leqslant \left(\frac{1}{C_n}\right)^{\frac{n}{n-1}}\sum_{j=l_1+1}^l|\delta O_j|^{\frac{n}{n-1}}\\
&\leqslant\left(\frac{1}{C_n}\right)^{\frac{n}{n-1}}\left(\sum_{j=l_1+1}^l|\delta O_j|\right)^{\frac{n}{n-1}}
\leqslant \left(\frac{2n}{C_n}\right)^{\frac{n}{n-1}}|A_3^i|^{\frac{n}{n-1}}\\
&\leqslant \left(\frac{2nN}{C_n}\right)^{\frac{n}{n-1}},
\end{aligned}
\end{equation*}
which contradicts the claim proved in $(ii)$ by letting $i\rightarrow+\infty$.

With this fact, we choose a small constant $\epsilon$ satisfying
$$0<\epsilon<\inf_{x\in[-(n+1)C-\lambda,-C]}e^x(1-e^x).$$
From (5), we get
\begin{equation*}
\begin{aligned}
\sum_{x\in \delta\Omega_i}\frac{\partial\widetilde{u^i}}{\partial \vec{n}}(x)+\lambda\sum_{x\in\Omega_i}e^{\widetilde{u^i}}(1-e^{\widetilde{u^i}})
=\sum_{x\in\Omega_i}g(x)=4\pi\sum_{j=1}^Mn_j=C,
\end{aligned}
\end{equation*}
which implies that
$$\sum_{x\in\Omega_i}e^{\widetilde{u^i}}(1-e^{\widetilde{u^i}})\leqslant \frac{C}{\lambda}$$
and
$$\sum_{x\in A_3^i}\epsilon\leqslant \frac{C}{\lambda}.$$
This yields a contradiction to $(iii)$. Thus $\{\widetilde{u^i}\}$ has a uniform bound in $l^\infty(\mathbb{Z}^n)$, and we have the pointwise convergence
$$\lim_{i\rightarrow+\infty}\widetilde{u^i}(x)=u(x),\ \forall x\in\mathbb{Z}^n,$$
where $u\in l^\infty(\mathbb{Z}^n)$ and satisfies the self-dual Chern-Simons vortex equation (1). From the above inequality and the fact that $u$ has a lower bound, we pass to the limit, and get that for any $i\geqslant1$,
$$e^{\inf_{x\in \mathbb{Z}^n}u(x)}\sum_{x\in\Omega_i}(1-e^u)\leqslant \sum_{x\in\Omega_i}e^u(1-e^u)\leqslant\frac{C}{\lambda},$$
which yields that $u$ is a topological solution.

The solution is maximal follows from Lemma~\ref{lm5}. On any finite subset $\Omega$, by Lemma~\ref{lm5}, we obtain that the solution $u_\Omega$ is maximal. On $\mathbb{Z}^n$, we suppose that there exists another topological solution $f$ of the self-dual Chern-Simons vortex equation. From the proof of Lemma~\ref{lm5}, we observe that $f\leqslant0$ on $\mathbb{Z}^n$. Applying Lemma~\ref{lm5} on $\Omega_i$, we have
$$f \leqslant u^i.$$
For a fixed integer $k\geqslant1$, and for $i\geqslant k$ we have
$$f(x) \leqslant \underline{\lim}_{i\rightarrow\infty}u^i(x)=u(x)\ \ \text{on}\ \Omega_k.$$
For any $x\in\mathbb{Z}^n$, there exists a sufficiently large integer $k$ satisfying $x\in \Omega_k$ such that $f(x)\leqslant u(x)$. Thus we obtain $f\leqslant u$ on $\mathbb{Z}^n$, and the solution $u$ is maximal among all possible solutions.

The last part is to prove the decay estimate
$$u= O(e^{-m(1-\epsilon) d(x)}),$$
where $m=\ln(1+\frac{\lambda}{2n})$.
Note that the solution $u$ satisfies
$$\Delta u=\lambda e^u(e^u-1)\ \ \text{on}\ \overline{\Omega_0}^c.$$
Since
$$\lim_{d(x)\rightarrow+\infty}u(x) =0,$$
for any $0<\epsilon<1$, we can choose $R\geqslant1$ sufficiently large such that
$$\lambda e^{2u}\geqslant 2n\left[\left(1+\frac{\lambda}{2n}\right)^{1-\epsilon}-1\right],\ \ d(x)\geqslant R.$$
Then for $d(x)\geqslant R$,
$$\Delta u=\lambda e^u(e^u-1)=\lambda e^{u+\omega}u\leqslant \lambda e^{2u}u \leqslant c_3 u,$$
where the function $\omega$ satisfies $u\leqslant\omega\leqslant 0$ and $c_3=2n\left[\left(1+\frac{\lambda}{2n}\right)^{1-\epsilon}-1\right]$.

Consider the function $h(x)= -e^{-m(1-\epsilon)d(x)}$. Let $e_i$ be the vector whose $i$-th component is 1 and the others are 0. For $x\in \overline{\Omega_0}^c$ and $d(x)\geqslant R$, suppose that $d(x)=t\geqslant R\geqslant1$, and we have
$$\Delta h(x)=\sum_{y\sim x}(h(y)-h(x))=\sum_{i=1}^n(h(x+e_i)+h(x-e_i)-2h(x)).$$
If $x_i\neq0$, then
$$h(x+e_i)+h(x-e_i)-2h(x)=-e^{-m(1-\epsilon)(t-1)}-e^{-m(1-\epsilon)(t+1)}+2e^{-m(1-\epsilon)t}.$$
If $x_i=0$, we have
\begin{equation*}
\begin{aligned}
h(x+e_i)+h(x-e_i)-2h(x)&=-2e^{-m(1-\epsilon)(t+1)}+2e^{-m(1-\epsilon)t}\\
&\geqslant -e^{-m(1-\epsilon)(t-1)}-e^{-m(1-\epsilon)(t+1)}+2e^{-m(1-\epsilon)t}.
\end{aligned}
\end{equation*}
Therefore, we have the following inequality
\begin{equation*}
\begin{aligned}
\Delta h(x)&\geqslant n\left[-e^{-m(1-\epsilon)(t-1)}-e^{-m(1-\epsilon)(t+1)}+2e^{-m(1-\epsilon)t}\right]\\
&=n\left[e^{-m(1-\epsilon)}+e^{m(1-\epsilon)}-2\right]h(x)\\
&=n\left[\left(1+\frac{c_3}{2n}\right)+\frac{1}{1+\frac{c_3}{2n}}-2\right]h(x)\\
&\geqslant n\left[2\left(1+\frac{c_3}{2n}\right)-2\right]h(x)=c_3h(x).
\end{aligned}
\end{equation*}
Fix a subset
$$\Omega_0'=\{x\in\mathbb{Z}^n: d(x)\geqslant R_1\geqslant R\},$$
which satisfies $\Omega_0'\cap \overline{\Omega_0}=\emptyset$. By choosing a large constant $C(\epsilon)$, we obtain
$$(\Delta-c_3)(C(\epsilon)h-u)\geqslant0\ \ \text{on}\ \Omega_0',$$
$$\lim_{|x|\rightarrow+\infty}(C(\epsilon)h-u)(x)=0,$$
and
$$C(\epsilon)h(x)-u(x)\leqslant 0 \ \ \text{if}\ d(x)=R_1.$$
These imply that
$$0\geqslant u(x)\geqslant -C(\epsilon) e^{-m(1-\epsilon)d(x)}\ \ \text{on}\ \Omega_0',$$
completing Theorem~\ref{thm:main2}. As a consequence, we obtain $u\in l^p(\mathbb{Z}^n)$, $1\leqslant p\leqslant\infty$.
\end{proof}

\

\subsection{The proof B of Theorem~\ref{thm:main2}}

\

The proof B follows the methods in \cite{spruck1995topological}. We mainly prove the following key lemma.
\begin{lemma}\label{norm}
Let $n\geqslant2,\ \lambda>0$, and $\Omega_0$ be a finite subset of $\mathbb{Z}^n$ containing the distinct points $\{p_j\}_{j=1}^M$. For any finite subset $\Omega\supset\Omega_0$, the boundary value problem
\begin{equation*}
\left\{
\begin{aligned}
& \Delta u=\lambda e^u(e^u-1)+4\pi\sum_{j=1}^Mn_j\delta_{p_j}\ \ \text{on} \ \Omega,\\
& u(x) =0\ \ \text{on} \ \delta\Omega\\
\end{aligned}
\right.
\end{equation*}
has a solution $u_{\Omega}:\overline{\Omega}\rightarrow\mathbb{R}$. This solution is maximal among all possible solutions and satisfies that $\|u_\Omega\|_{l^2(\Omega)}\leqslant C_0$, where $C_0$ only depends on $n$, $\lambda$ and $C$.
\end{lemma}

By Green's identities in Lemma~\ref{lm1}, we consider the following functional on $\Omega$
$$F(u)=\frac{1}{2}D_\Omega(u)+\sum_{x\in\Omega}\left[\frac{\lambda}{2}(e^{u(x)}-1)^2+g(x)u(x)\right].$$
We prove the following lemma which states that $F(u_k)$ deceases with respect to $k$.

\begin{lemma}\label{lm6}
Let $\{u_k\}$ be the sequence defined by $(3)$. Then
$$c_0\geqslant F(u_1)\geqslant F(u_2)\geqslant\ldots\geqslant F(u_k)\geqslant\ldots,$$
where the constant $c_0$ only depends on $n, C, \lambda$.
\end{lemma}

\begin{proof}
Multiplying $(3)$ by $u_k-u_{k-1}$ and summing over $\Omega$, we obtain
\begin{equation}
\begin{split}
&\sum_{x\in\Omega}(\Delta-K)u_k(x)\left[u_k(x)-u_{k-1}(x)\right]\\
=&\sum_{x\in\Omega}[\lambda e^{u_{k-1}}(e^{u_{k-1}}-1)(u_k-u_{k-1})-
Ku_{k-1}(u_k-u_{k-1})+g(u_k-u_{k-1})](x).
\end{split}
\end{equation}
By Green's identities in Lemma~\ref{lm1},
\begin{equation*}
\begin{aligned}
\sum_{x\in\Omega}\Delta u_k(x)(u_k(x)-u_{k-1}(x))=-D_\Omega(u_k-u_{k-1},u_k)=-D_\Omega(u_k)+D_\Omega(u_{k-1},u_k).\\
\end{aligned}
\end{equation*}
Combining it with the equation $(6)$, we get
\begin{equation*}
\begin{aligned}
&D_\Omega(u_k)-D_\Omega(u_{k-1},u_k)+\sum_{x\in\Omega}K(u_k(x)-u_{k-1}(x))^2\\
=&-\sum_{x\in\Omega}[\lambda e^{u_{k-1}}(e^{u_{k-1}}-1)(u_k-u_{k-1})+g(u_k-u_{k-1})](x).
\end{aligned}
\end{equation*}

Consider the function
$$\varphi (x)=\frac{\lambda}{2}(e^x-1)^2-\frac{K}{2}x^2,$$
which is concave for  any $x\leqslant 0$. Hence
\begin{equation*}
\begin{aligned}
\frac{\varphi(u_{k-1})-\varphi(u_{k})}{u_{k-1}-u_k}\geqslant \varphi'(u_{k-1})=\lambda e^{u_{k-1}}(e^{u_{k-1}}-1)-Ku_{k-1}.
\end{aligned}
\end{equation*}
That is,
$$\frac{\lambda}{2}(e^{u_k}-1)^2\leqslant \frac{\lambda}{2}(e^{u_{k-1}}-1)^2+\frac{K}{2}(u_k-u_{k-1})^2+\lambda e^{u_{k-1}}(e^{u_{k-1}}-1)(u_k-u_{k-1}).$$
By the fact
\begin{equation*}
\begin{aligned}
&|D_\Omega(u_{k-1},u_k)|\leqslant \frac{1}{2}\sum_{\substack{x,y\in \Omega\\ x\sim y}}|\nabla_{xy}u_{k-1}\nabla_{xy}u_{k}|+\sum_{\substack{x\in \Omega,y\in \delta\Omega\\ x\sim y}}|\nabla_{xy}u_{k-1}\nabla_{xy}u_{k}|\\
\leqslant& \frac{1}{4}\sum_{\substack{x,y\in \Omega\\ x\sim y}}(|\nabla_{xy}u_{k-1}|^2+|\nabla_{xy}u_{k}|^2)+\frac{1}{2}\sum_{\substack{x\in \Omega,y\in \delta\Omega\\ x\sim y}}(|\nabla_{xy}u_{k-1}|^2+|\nabla_{xy}u_{k}|^2)\\
=& \frac{1}{2}D_\Omega(u_{k-1})+\frac{1}{2}D_\Omega(u_{k}),
\end{aligned}
\end{equation*}
we obtain that
$$F(u_k)\leqslant F(u_k)+\frac{K}{2}\|u_{k-1}-u_k\|^2_{l^2(\Omega)}\leqslant F(u_{k-1}).$$
Thus we only need to prove $F(u_1)\leqslant c_0$. Note that
\begin{equation*}
\begin{aligned}
D_\Omega(u_{1})&=\frac{1}{2}\sum_{\substack{x,y\in \Omega\\ x\sim y}}|\nabla_{xy}u_{1}|^2+\sum_{\substack{x\in \Omega,y\in \delta\Omega\\ x\sim y}}|\nabla_{xy}u_{1}|^2\\
&\leqslant\sum_{\substack{x,y\in \Omega\\ x\sim y}}(u_1(x)^2+u_1(y)^2)+2\sum_{\substack{x\in \Omega,y\in \delta\Omega\\ x\sim y}}(u_1(x)^2+u_1(y)^2)\\
&\leqslant 4n \|u_1\|_{l^2(\Omega)}^2,
\end{aligned}
\end{equation*}
and $|e^{u_1}-1|=1-e^{u_1}\leqslant -u_1$. Then we have the estimate
\begin{equation*}
\begin{aligned}
F(u_1)&\leqslant \frac{1}{2}\cdot4n\|u_1\|_{l^2(\Omega)}^2+\frac{\lambda}{2}\sum_{x\in\Omega} u_1(x)^2+\frac{1}{2}\sum_{x\in\Omega} [g(x)^2+u_1(x)^2]\\
&=c_1+c_2\|u_1\|_{l^2(\Omega)}^2,
\end{aligned}
\end{equation*}
where $c_1,c_2$ are constants that only depend on $n$, $\lambda$ and $C$. Multiplying $(4)$ by $u_1$ and summing over $\Omega$, we have
\begin{equation*}
\begin{aligned}
D_\Omega(u_1)+K\sum_{x\in\Omega} u_1(x)^2=-\sum_{x\in\Omega} g(x)u_1(x).
\end{aligned}
\end{equation*}
This yields
$$K\sum_{x\in\Omega} u_1(x)^2\leqslant \frac{1}{2K}\sum_{x\in\Omega}g(x)^2+\frac{K}{2}\sum_{x\in\Omega}u_1(x)^2.$$
Hence,
$$\sum_{x\in\Omega} u_1(x)^2\leqslant \frac{\|g\|_{l^2(\mathbb{Z}^n)}^2}{K^2},$$
which completes the proof.
\end{proof}

\

Our aim is a uniform control of the $l^2$-norm of $\{u_k\}$. By Lemma~\ref{lm6}, we can use the functional $F(u_k)$ to control the $l^2$-norm of $u_k$. In fact, we prove the following lemma, which states that the functional $F$ is coercive.

\begin{lemma}\label{lm7}
Let $v\in l^2(\overline{\Omega})$ and $v(x)=0$ for all $x\in \delta\Omega$. Then
$$\|v\|_{l^2(\Omega)}\leqslant C_2(F(v)+1),$$
where $C_2$ only depends on $n, C, \lambda$. In particular, let $\{u_k\}$ be the sequence defined by $(3)$. We have for any $k\geqslant1$,
$$\|u_k\|_{l^2(\Omega)}\leqslant C_2(F(u_k)+1)\leqslant C_0,$$
where $C_0$ only depends on $n, C, \lambda$.
\end{lemma}

\begin{proof}
For any function $v\in l^2(\overline{\Omega})$ with $v(x)=0$ for all $x\in \delta\Omega$. Let $\widetilde{v}$ be the null extension to $\mathbb{Z}^n$ of $v$ on $\Omega$.
Hence $\widetilde{v}\in l^2(\mathbb{Z}^n)$. By Lemma~\ref{GN}, we have
$$\|\widetilde{v}\|^4_{l^4(\mathbb{Z}^n)}\leqslant C_n' \|\widetilde{v}\|^2_{D^{1,2}(\mathbb{Z}^n)}\|\widetilde{v}\|^2_{l^2(\mathbb{Z}^n)}.$$
Note that
$$\|\widetilde{v}\|^4_{l^4(\mathbb{Z}^n)}=\sum_{x\in\Omega} v(x)^4,$$
$$\|\widetilde{v}\|^2_{l^2(\mathbb{Z}^n)}=\sum_{x\in\Omega}v(x)^2,$$
and
$$\|\widetilde{v}\|_{D^{1,2}(\mathbb{Z}^n)}\leqslant (2D_\Omega(v))^\frac{1}{2}.$$
This yields that
\begin{equation}
\begin{aligned}
\sum_{x\in\Omega} v(x)^4\leqslant C_3 D_\Omega(v)\sum_{x\in\Omega}v(x)^2,
\end{aligned}
\end{equation}
where $C_3=2C_n'$. Since $e^v-1\geqslant v$ and $1-e^{-v}\geqslant \frac{v}{1+v}$ for $v\geqslant0$,
$$|e^v-1|^2\geqslant \left(\frac{|v|}{1+|v|}\right)^2.$$
By $(7)$ we obtain
\begin{equation}
\begin{split}
F(v)&=\frac{1}{2}D_\Omega(v)+\sum_{x\in\Omega}\left[\frac{\lambda}{2}(e^{v(x)}-1)^2+g(x)u(x)\right]\\
&\geqslant \frac{1}{2}D_\Omega(v)+\frac{\lambda}{2}\sum_{x\in\Omega}\left(\frac{|v(x)|}{1+|v(x)|}\right)^2-\|g\|_{l^{\frac{4}{3}}(\mathbb{Z}^n)}\|v\|_{l^4(\Omega)}\\
&\geqslant \frac{1}{2}D_\Omega(v)+\frac{\lambda}{2}\sum_{x\in\Omega}\left(\frac{|v(x)|}{1+|v(x)|}\right)^2- C_4(D_\Omega(v))^\frac{1}{4}\left(\sum_{x\in\Omega}v(x)^2\right)^{\frac{1}{4}}\\
&\geqslant \frac{1}{2}D_\Omega(v)+\frac{\lambda}{2}\sum_{x\in\Omega}\left(\frac{|v(x)|}{1+|v(x)|}\right)^2- \epsilon \|v\|_{l^2(\Omega)}- \frac{C_4}{\epsilon}(D_\Omega(v))^\frac{1}{2}\\
&\geqslant \frac{1}{2}D_\Omega(v)+\frac{\lambda}{2}\sum_{x\in\Omega}\left(\frac{|v(x)|}{1+|v(x)|}\right)^2- \epsilon \|v\|_{l^2(\Omega)}- \frac{1}{4}D_\Omega(v)-C_4\\
&=\frac{1}{4}D_\Omega(v)+\frac{\lambda}{2}\sum_{x\in\Omega}\left(\frac{|v(x)|}{1+|v(x)|}\right)^2- \epsilon \|v\|_{l^2(\Omega)}-C_4,
\end{split}
\end{equation}
where $\epsilon>0$ is a sufficient small constant which will be chosen below, and $C_4$ is a uniform constant only depending on $\epsilon, C, C_n'$ which may change its value from line to line.

By the inequality $(7)$, we have the following estimate
\begin{align*}
&\left(\sum_{x\in\Omega}v(x)^2\right)^2=\left[ \sum_{x\in\Omega}\frac{|v(x)|}{1+|v(x)|}(1+|v(x)|)|v(x)| \right]^2 \\
\leqslant &\sum_{x\in\Omega}\left(\frac{|v(x)|}{1+|v(x)|}\right)^2 \sum_{x\in\Omega}\left(1+|v(x)|\right)^2v(x)^2\\
\leqslant &2\sum_{x\in\Omega}\left(\frac{|v(x)|}{1+|v(x)|}\right)^2 \sum_{x\in\Omega}\left(v(x)^2+v(x)^4\right)\\
\leqslant &2\sum_{x\in\Omega}\left(\frac{|v(x)|}{1+|v(x)|}\right)^2\sum_{x\in\Omega}v(x)^2+2C_3\sum_{x\in\Omega}\left(\frac{|v(x)|}{1+|v(x)|}\right)^2D_\Omega(v)\sum_{x\in\Omega}v(x)^2\\
\leqslant &\frac{1}{2}\left(\sum_{x\in\Omega}v(x)^2\right)^2+C_5\left[\left(\sum_{x\in\Omega}\left(\frac{|v(x)|}{1+|v(x)|}\right)^2\right)^2
+\left(\sum_{x\in\Omega}\left(\frac{|v(x)|}{1+|v(x)|}\right)^2\right)^2D_\Omega(v)^2\right]\\
\leqslant &\frac{1}{2}\left(\sum_{x\in\Omega}v(x)^2\right)^2+C_5\left[1+\left(\sum_{x\in\Omega}\left(\frac{|v(x)|}{1+|v(x)|}\right)^2\right)^4+D_\Omega(v)^4\right],\\
\end{align*}
where $C_5,C_6$ are uniform constant only depending on $C_n'$. This yields that
\begin{equation}
\begin{aligned}
\|v\|_{l^2(\Omega)} \leqslant C_6\left[1+\sum_{x\in\Omega}\left(\frac{|v(x)|}{1+|v(x)|}\right)^2+D_\Omega(v)\right].
\end{aligned}
\end{equation}
We choose $\epsilon = \frac{\min\{\frac{1}{8},\frac{\lambda}{4}\}}{C_6}$, and by combining $(8)$ with $(9)$, we obtain
$$\|v\|_{l^2(\Omega)}\leqslant C_2(F(v)+1).$$
By Lemma~\ref{lm6}, we have
$$\|u_k\|_{l^2(\Omega)}\leqslant C_2(F(u_k)+1)\leqslant C_0,$$
where $C_2,C_0$ only depend on $n, C, \lambda$.
\end{proof}

\

\begin{proof}[Proof of Lemma~\ref{norm}]
By Lemma~\ref{lm7} and Lemma~\ref{lm3}, we obtain
\begin{equation*}
\begin{aligned}
u_k\rightarrow\ u_\Omega \ \ &\text{in}\ l^2(\Omega),
\end{aligned}
\end{equation*}
and
$$\|u_\Omega\|_{l^2(\Omega)}\leqslant C_0.$$
Since $\Delta$ is a local operator, by the pointwise convergence the function $u_\Omega\in l^2(\Omega)$ is the solution to the equation
\begin{equation*}
\left\{
\begin{aligned}
& \Delta u=\lambda e^u(e^u-1)+g\ \ \text{on} \ \Omega,\\
& u(x) =0\ \ \text{on} \ \delta\Omega.\\
\end{aligned}
\right.
\end{equation*}
This finishes the main proof of Lemma~\ref{norm}. For the rest, it remains to prove that this solution is maximal, which we argue the same as in the proof A.

\end{proof}

\

Let $\Omega_i$ be finite and connected subsets, satisfying
$$\Omega_0\subset \Omega_1\subset\ldots\subset \Omega_k\subset\ldots,\ \ \bigcup_{i=1}^\infty\Omega_i=\mathbb{Z}^n,$$
and we write $u^i =u_{\Omega_i}$. Finally we use these lemmas to prove Theorem~\ref{thm:main2}.

\begin{proof}[Proof of Theorem~\ref{thm:main2}]
As in the proof A, for any integer $1\leqslant j\leqslant k$, one easily sees that
$$u^k\leqslant u^j \ \ \text{on}\ \overline{\Omega_j}.$$
Let $\widetilde{u^k}$ be the null extension to $\mathbb{Z}^n$ of $u_k$.
Then
$$0\geqslant \widetilde{u^1}\geqslant\widetilde{u^2}\geqslant\ldots\geqslant\widetilde{u^k}\geqslant\ldots$$
on $\mathbb{Z}^n$.
Noting that $\|\widetilde{u^k}\|_{l^2(\mathbb{Z}^n)}\leqslant C_0$ for any $k\geqslant1$, we have the pointwise convergence
\begin{equation*}
\begin{aligned}
\widetilde{u^k}(x)\rightarrow\ u(x), \ \forall x\in\mathbb{Z}^n ,
\end{aligned}
\end{equation*}
and $u\in l^2(\mathbb{Z}^n)$. Hence, $u$ satisfies the equations
\begin{equation*}
\left\{
\begin{aligned}
& \Delta u=\lambda e^u(e^u-1)+4\pi\sum_{j=1}^M\delta_{p_j}\ \ \text{on} \ \mathbb{Z}^n,\\
& \lim_{d(x)\rightarrow+\infty}u(x) =0,\\
\end{aligned}
\right.
\end{equation*}
which is a topological solution. Analogous to the proof A, one can show that this solution is maximal and satisfies the decay estimate. This implies that $u\in l^p(\mathbb{Z}^n)$ for any $1\leqslant p\leqslant\infty$, and we finish the proof B.

\end{proof}

\

\section{Existence theorems of the Abelian Higgs equation}\label{sec:3}

Note that the topological solution to the Chern-Simons model, obtained in Theorem~\ref{thm:main2}, serves as a subsolution of the Abelian Higgs equation. In this section we will prove the existence of topological solutions to the Abelian Higgs equation $(2)$ on $\mathbb{Z}^n$ for $n\geqslant2$ using the sub-supersolution approach. We prove the existence of topological solutions to $(2)$ by a monotone iteration method.

%The proof strategy is similar to Section~\ref{sec:2}, and in the following part, we will describe the differences in the proofs of two equations, and omit some proof steps which are the same as Section~\ref{sec:2}.

\begin{definition}
We call a function $\omega$ supersolution (resp. subsolution) of $(2)$, if
$$\Delta \omega\leqslant (resp. \geqslant)\ \lambda (e^\omega-1)+g\ \ \text{on} \ \mathbb{Z}^n.$$
\end{definition}

\

\begin{lemma}\label{lm8}
The function $\omega_1=0$ is a supersolution of $(2)$, and the function $\omega_2=u$ given by Theorem~\ref{thm:main2} is a subsolution of $(2)$.
\end{lemma}

\begin{proof}
Since $g\geqslant0$, we have
$$\Delta \omega_1=0\leqslant \lambda (e^{\omega_1}-1)+g.$$
Noting that $\omega_2\leqslant 0$, we have
$$\Delta \omega_2 = \lambda e^{\omega_2}(e^{\omega_2}-1)+g\geqslant\lambda (e^{\omega_2}-1)+g.$$
\end{proof}
%Let $\Omega_0$ be a finite subset of $\mathbb{Z}^n$, which satisfies $\Omega_0\supset \{p_j\}_{j=1}^M$, and let $\Omega$ be an arbitrary finite subset, satisfying $\Omega_0\subset\Omega\subset\mathbb{Z}^n$. Consider the equations
%\begin{equation}
%\left\{
%\begin{aligned}
%& \Delta u=\lambda (e^u-1)+g\ \ \text{on} \ \Omega,\\
%& u(x) =0\ \ \text{on} \ \delta\Omega,\\
%\end{aligned}
%\right.
%\end{equation}
%and the functional on $\Omega$
%$$G(u)=\frac{1}{2}D_\Omega(u)+\sum_{x\in\Omega}[\lambda(e^{u(x)}-u(x)-1)+g(x)u(x)].$$

Similar to Section~\ref{sec:2}, we define an iterative sequence as follows. For $K>\lambda$, let $u'_0=0$ and consider the following equations, $k\geqslant1$,
\begin{equation}
\left\{
\begin{aligned}
& (\Delta-K) u'_k=\lambda (e^{u'_{k-1}}-1)+g-Ku'_{k-1}\ \ \text{on} \ \mathbb{Z}^n,\\
& \lim_{d(x)\rightarrow+\infty}u'_k(x) =0.\\
\end{aligned}
\right.
\end{equation}
We have the following lemma.

\begin{lemma}\label{lm9}
Let $\{u'_k\}$ be the sequence defined by $(10)$. Then for each $k$, $u'_k$ is uniquely defined and
$$\omega_1=0=u'_0\geqslant u'_1\geqslant u'_2\geqslant\ldots\geqslant \omega_2.$$
\end{lemma}

\begin{proof}
It is clear that $u'_1$ is unique and $u'_1\in l^2(\mathbb{Z}^n)$. Since
\begin{equation*}
\begin{aligned}
(\Delta-K)(\omega_2-u'_1)\geqslant \lambda(e^{\omega_2}-1)-K\omega_2\geqslant (\lambda-K)\omega_2\geqslant0,
\end{aligned}
\end{equation*}
with the boundary conditions, we prove that $u'_1\geqslant \omega_2$ by Corollary~\ref{co1}.

Suppose that
$$0=u'_0\geqslant u'_1\geqslant u'_2\geqslant\ldots\geqslant u'_i\geqslant \omega_2,$$
and we have the existence and uniqueness of $u'_{i+1}\in l^2(\mathbb{Z}^n)$. By calculation, we obtain
\begin{equation*}
\begin{aligned}
(\Delta-K)(u'_{i+1}-u'_i)&=\lambda(e^{u'_i}-e^{u'_{i-1}})-K(u'_i-u'_{i-1})\\
&\geqslant\lambda e^{\eta_1}(u'_i-u'_{i-1})-K(u'_i-u'_{i-1})\\
&\geqslant K(e^{\eta_1}-1)(u'_i-u'_{i-1})\geqslant 0,
\end{aligned}
\end{equation*}
and
\begin{equation*}
\begin{aligned}
(\Delta-K)(\omega_2-u'_{i+1})&\geqslant\lambda(e^{\omega_2}-e^{u'_i})-K(\omega_2-u'_i)\\
&\geqslant(\lambda e^{\eta_2}-K)(\omega_2-u'_i)\geqslant0\\
\end{aligned}
\end{equation*}
where the functions $\eta_1,\eta_2$ satisfy
$$u'_i\leqslant\eta_1\leqslant u'_{i-1}\leqslant0, \ \omega_2\leqslant\eta_2\leqslant u'_i\leqslant0.$$
These yield that
$$\omega_2\leqslant u'_{i+1}\leqslant u'_i.$$

\end{proof}

Finally, we give a sketch of the proof of Theorem~\ref{thm:main3}.

\begin{proof}[Proof of Theorem~\ref{thm:main3}]
By Lemma~\ref{lm9}, the monotone sequence $\{u'_k\}$ is bounded in $l^2(\mathbb{Z}^n)$. Hence, we get the pointwise convergence
$$u'_k(x)\rightarrow u'(x),\ \forall x\in\mathbb{Z}^n,$$
and we obtain that $u'\in l^2(\mathbb{Z}^n)$ and $u'$ is a topological solution to $(2)$. In addition, if there exists another topological solution $f$, then
$$\Delta(u'-f)=\lambda(e^{u'}-e^f).$$
Hence, there exists a function $f'$, satisfying $\min\{u', f\}\leqslant f'\leqslant\max\{u', f\}$, such that
$$(\Delta - \lambda e^{f'})(u'-f)=0.$$
By the maximum principle, we obtain that this solution is unique.

Furthermore, since $0\geqslant u'\geqslant \omega_2=u$ and
$$u= O(e^{-m(1-\epsilon) d(x)}),$$
we have
$$u'= O(e^{-m(1-\epsilon) d(x)}),$$
and $u'\in l^p(\mathbb{Z}^n)$ for any $1\leqslant p\leqslant\infty$.
\end{proof}
\bigskip
\bigskip

\bigskip

\bibliographystyle{alpha}
\bibliography{ckwx}

\begin{thebibliography}{CHMY94}

\bibitem[Bar17]{barlow2017random}
Martin~T Barlow.
\newblock {\em Random walks and heat kernels on graphs}, volume 438.
\newblock Cambridge University Press, 2017.

\bibitem[CFL02]{chan2002non}
Hsungrow Chan, Chun-Chieh Fu, and Chang-Shou Lin.
\newblock Non-topological multi-vortex solutions to the self-dual
  chern-simons-higgs equation.
\newblock {\em Communications in mathematical physics}, 231(2):189--221, 2002.

\bibitem[CH23]{chao2023multiple}
Ruixue Chao and Songbo Hou.
\newblock Multiple solutions for a generalized chern-simons equation on graphs.
\newblock {\em Journal of Mathematical Analysis and Applications},
  519(1):126787, 2023.

\bibitem[CHMY94]{chen1994nonlinear}
Xinfu Chen, Stuart Hastings, John~Bryce McLeod, and Yisong Yang.
\newblock A nonlinear elliptic equation arising from gauge field theory and
  cosmology.
\newblock {\em Proceedings of the Royal Society of London. Series A:
  Mathematical and Physical Sciences}, 446(1928):453--478, 1994.

\bibitem[CI00]{chae2000existence}
Dongho Chae and Oleg~Yu Imanuvilov.
\newblock The existence of non-topological multivortex solutions in the
  relativistic self-dual chern--simons theory.
\newblock {\em Communications in Mathematical Physics}, 215:119--142, 2000.

\bibitem[CK97]{chae1997topological}
Dongho Chae and Namkwon Kim.
\newblock Topological multivortex solutions of the self-dual
  maxwell--chern--simons--higgs system.
\newblock {\em journal of differential equations}, 134(1):154--182, 1997.

\bibitem[CKL11]{choe2011existence}
Kwangseok Choe, Namkwon Kim, and Chang-Shou Lin.
\newblock Existence of self-dual non-topological solutions in the chern--simons
  higgs model.
\newblock In {\em Annales de l'Institut Henri Poincar{\'e} C, Analyse non
  lin\'{e}aire}, volume~28, pages 837--852. Elsevier, 2011.

\bibitem[CY95]{caffarelli1995vortex}
Luis~A Caffarelli and Yisong Yang.
\newblock Vortex condensation in the chern-simons higgs model: an existence
  theorem.
\newblock {\em Communications in mathematical physics}, 168:321--336, 1995.

\bibitem[Dun09]{dunne2009self}
Gerald Dunne.
\newblock {\em Self-Dual Chern-Simons Theories}.
\newblock Springer Science \& Business Media, 2009.

\bibitem[GJ17]{article}
Huabin Ge and Wenfeng Jiang.
\newblock Kazdan-warner equation on infinite graphs.
\newblock {\em Journal of the Korean Mathematical Society}, 55, 06 2017.

\bibitem[GLY16a]{grigor2016kazdan}
Alexander Grigor'yan, Yong Lin, and Yunyan Yang.
\newblock Kazdan--warner equation on graph.
\newblock {\em Calculus of Variations and Partial Differential Equations},
  55:1--13, 2016.

\bibitem[GLY16b]{grigor2016yamabe}
Alexander Grigor'yan, Yong Lin, and Yunyan Yang.
\newblock Yamabe type equations on graphs.
\newblock {\em Journal of Differential Equations}, 261(9):4924--4943, 2016.

\bibitem[GLY17]{grigor2017existence}
Alexander Grigor'yan, Yong Lin, and YunYan Yang.
\newblock Existence of positive solutions to some nonlinear equations on
  locally finite graphs.
\newblock {\em Science China Mathematics}, 60:1311--1324, 2017.

\bibitem[Gri18]{grigor2018introduction}
Alexander Grigor'yan.
\newblock {\em Introduction to analysis on graphs}, volume~71.
\newblock American Mathematical Soc., 2018.

\bibitem[Han14]{han2014existence}
Xiaosen Han.
\newblock Existence of doubly periodic vortices in a generalized chern--simons
  model.
\newblock {\em Nonlinear Analysis: Real World Applications}, 16:90--102, 2014.

\bibitem[HKP90]{hong1990multivortex}
Jooyoo Hong, Yoonbai Kim, and Pong~Youl Pac.
\newblock Multivortex solutions of the abelian chern-simons-higgs theory.
\newblock {\em Physical Review Letters}, 64(19):2230, 1990.

\bibitem[HLY13]{GLY}
Genggeng Huang, Congming Li, and Ximing Yin.
\newblock Existence of the maximizing pair for the discrete
  hardy-littlewood-sobolev inequality.
\newblock {\em Discrete and Continuous Dynamical Systems}, 35, 09 2013.

\bibitem[HLY20]{huang2020existence}
An~Huang, Yong Lin, and Shing-Tung Yau.
\newblock Existence of solutions to mean field equations on graphs.
\newblock {\em Communications in mathematical physics}, 377(1):613--621, 2020.

\bibitem[HS22]{hou2022existence}
Songbo Hou and Jiamin Sun.
\newblock Existence of solutions to chern--simons--higgs equations on graphs.
\newblock {\em Calculus of Variations and Partial Differential Equations},
  61(4):139, 2022.

\bibitem[HWY21]{huang2021mean}
Hsin-Yuan Huang, Jun Wang, and Wen Yang.
\newblock Mean field equation and relativistic abelian chern-simons model on
  finite graphs.
\newblock {\em Journal of Functional Analysis}, 281(10):109218, 2021.

\bibitem[JR79]{jacobs1979interaction}
Laurence Jacobs and Claudio Rebbi.
\newblock Interaction energy of superconducting vortices.
\newblock {\em Physical review B}, 19(9):4486, 1979.

\bibitem[JT80]{jaffe1980vortices}
A.~Jaffe and C.~Taubes.
\newblock {\em Vortices and Monopoles: Structure of Static Gauge Theories}.
\newblock Progress in physics. Birkh\"{a}user, 1980.

\bibitem[JW90]{jackiw1990self}
Roman Jackiw and Erick~J Weinberg.
\newblock Self-dual chern-simons vortices.
\newblock {\em Physical Review Letters}, 64(19):2234, 1990.

\bibitem[KS18]{article1}
Matthias Keller and Michael Schwarz.
\newblock The kazdan-warner equation on canonically compactifiable graphs.
\newblock {\em Calculus of Variations and Partial Differential Equations}, 57,
  03 2018.

\bibitem[LSY23]{li2023topological}
Jiayu Li, Linlin Sun, and Yunyan Yang.
\newblock Topological degree for chern-simons higgs models on finite graphs.
\newblock {\em arXiv preprint arXiv:2309.12024}, 2023.

\bibitem[LY22]{lin2022calculus}
Yong Lin and Yunyan Yang.
\newblock Calculus of variations on locally finite graphs.
\newblock {\em Revista Matem{\'a}tica Complutense}, pages 1--23, 2022.

\bibitem[LZ21]{lu2021existence}
Yingshu L\"{u} and Peirong Zhong.
\newblock Existence of solutions to a generalized self-dual chern-simons
  equation on graphs.
\newblock {\em arXiv preprint arXiv:2107.12535}, 2021.

\bibitem[Por20]{porretta2020note}
Alessio Porretta.
\newblock {A Note on the Sobolev and Gagliardo--Nirenberg Inequality when $p >
  N$}.
\newblock {\em Advanced Nonlinear Studies}, 20(2):361--371, 2020.

\bibitem[ST98]{struwe1998multivortex}
Michael Struwe and Gabriella Tarantello.
\newblock On multivortex solutions in chern-simons gauge theory.
\newblock {\em Bollettino della Unione Matematica Italiana-B}, 1:109--122,
  1998.

\bibitem[SY95]{spruck1995topological}
Joel Spruck and Yisong Yang.
\newblock Topological solutions in the self-dual chern-simons theory: existence
  and approximation.
\newblock In {\em Annales de l'Institut Henri Poincar\'{e} C, Analyse non
  lin\'{e}aire}, volume~12, pages 75--97. Elsevier, 1995.

\bibitem[Wan91]{ronggang1991existence}
Ronggang Wang.
\newblock The existence of chern-simons vortices.
\newblock {\em Communications in Mathematical Physics;(Germany, FR)}, 137(3),
  1991.

\bibitem[WY92]{wang1992abrikosov}
Sheng Wang and Yisong Yang.
\newblock Abrikosov's vortices in the critical coupling.
\newblock {\em SIAM journal on mathematical analysis}, 23(5):1125--1140, 1992.

\end{thebibliography}

\end{document}